\documentclass{llncs}

\usepackage{graphicx}
\usepackage{amsmath}

\title{Lower Bounds in the Preprocessing and Query Phases of
 Routing Algorithms}

\author{Colin White}

\institute{Carnegie Mellon University \\
	\email{crwhite@cs.cmu.edu}\qquad
	\texttt{http://cs.cmu.edu/{\raise.17ex\hbox{$\scriptstyle\sim$}}crwhite}
	}
\begin{document}
\maketitle

\begin{abstract}

In the last decade, there has been a substantial amount of research in finding
routing algorithms designed specifically to run on real-world graphs.
In 2010, Abraham et al.\ showed upper bounds
on the query time in
terms of a graph's \textit{highway dimension} and diameter for 
the current fastest routing
algorithms, including {\sc contraction hierarchies}, 
{\sc transit node routing}, and {\sc hub labeling}.
In this paper, we show corresponding lower bounds for the same three algorithms.
We also show how to improve a result by Milosavljevi\'c
which lower bounds the number of shortcuts added in the preprocessing
stage for {\sc contraction hierarchies}. We relax the assumption of
an optimal contraction order (which is NP-hard to compute), 
allowing the result to be applicable to real-world instances.
Finally, we give a proof that optimal preprocessing for 
{\sc hub labeling} is NP-hard. Hardness of optimal preprocessing
is known for most routing algorithms, 
and was suspected to be true for {\sc hub labeling}. 
\end{abstract}

\section{Introduction}\label{sec:Introduction}

The problem of finding shortest paths in road networks has been well-studied in the last
decade, motivated by the application of computing driving directions.  
Although Dijkstra's algorithm runs in small polynomial time, for applications involving
continental-sized road networks, Dijkstra's algorithm
is simply not fast enough.  There have been many different approaches
to find algorithms that specifically run fast on real-world graphs.

Most recent innovations involve a two-stage algorithm: a preprocessing stage and a query
stage.  The preprocessing stage runs once and can spend hours calculating data.
Then the query stage uses this data to find shortest paths very fast, often
several orders of magnitude faster than Dijkstra's algorithm for a continental query.  
Once the preprocessing
stage is completed, the users can run as many queries as they want.
For a query between two nodes $s$ and $t$ (an $s$--$t$ query), the algorithm returns
$\mbox{dist}(s,t)$, the cost of the shortest path between $s$ and $t$.  Most algorithms
can also return the vertices on the shortest path using an extra data structure.

The current fastest routing algorithm on real-world graphs is
{\sc hub labeling} \cite{hl}, which achieves a speedup of six orders of magnitude
over Dijkstra's algorithm.  The {\sc transit node routing} algorithm 
is 
second-fastest, and requires an order of magnitude less space than 
{\sc hub labeling}. {\sc contraction hierarchies} is also a fast routing
algorithm, which was state of the art in 2008.
For a comprehensive overview of the best routing algorithms, see \cite{rptn}.

Until recently, it was known that these algorithms performed very well on
real-world maps, but there were no theoretical guarantees.
In fact, it is not hard to construct specific graphs for which these algorithms
perform no faster than Dijkstra's algorithm.
So, an interesting theoretical question is to find properties present in all
real-life graphs that explain why these algorithms work so well.

With this motivation in mind,
Abraham et al.\ defined the notion of \textit{highway dimension} \cite{hdnew},
intuitively, the extent to which all shortest paths are 
hit by at least one of a small set of \textit{access nodes}.  
Although it is too computationally intensive to calculate the exact highway dimension
for a continental road map, there is evidence that the highway dimension $h$
is at most polylogarithmic in the number of vertices.
It is conjectured
that real-world routing networks always have low highway dimension, based on experimental evidence
\cite{hd}.
Abraham et al.\ were able to prove strong upper bounds on the query times in terms of highway dimension and
diameter $d$ for four of the fastest routing algorithms: 
{\sc hub labeling}, {\sc contraction hierarchies}, {\sc transit node routing},
and {\sc reach}.

\subsection{Our results}
In this paper, we are interested in finding lower bounds for the current
state-of-the-art routing algorithms.
We show tight or near-tight bounds on the runtime for
{\sc hub labeling}, {\sc contraction hierarchies}, and {\sc transit node routing}.

Our lower bounds may facilitate proving better guarantees of these
algorithms, or provide intuition for new routing algorithms, if
one can find differences between the graphs we use and real world
instances. For example, the graphs we use have low highway dimension,
but they do not have small separators and are nonplanar, so perhaps
there is a way to modify {\sc hub labeling} to take this into account.

We show a tight lower bound for {\sc hub labeling}, the fastest routing
algorithm to date \cite{rptn}.  
For {\sc contraction hierarchies} and
{\sc transit node routing}, the definition of highway dimension in the lower bound versus upper
bound is slightly different (because of a recent redefinition by Abraham et al.\ ), so we cannot quite say the bounds are tight.

We can also use our analysis to generalize a known result 
by Milosavljevi\'c,
which lower bounds the
number of shortcut edges in the preprocessing stage of 
{\sc contraction hierarchies} \cite{milo}.
This result assumes an optimal contraction order which is NP-hard to compute
\cite{hard}.
So for real-world instances, we rely on using contraction orders based
on heuristics.
We show how to relax the assumption about the contraction order,
which means the result can be applied to real-world instances.

We also contribute a hardness result for
optimal preprocessing of {\sc hub labeling}.
In 2010, Bauer et al.\ established hardness for optimal preprocessing for
a variety of the best routing algorithms, including 
{\sc contraction hierarchies} and {\sc transit node routing}.
In this paper, we show that in {\sc hub labeling} preprocessing,
the problem of minimizing the maximum label size
over all vertices is NP-hard.

This paper will proceed as follows.  
Section 2 will provide preliminary information,
specifically about highway dimension, and also the graph 
construction used in our
main theorems. In Section 3, we show a lower bound on the 
query time of the {\sc hub labeling} algorithm, and prove 
that optimal preprocessing is NP-hard.
In Section 4, we establish a lower bound on the query time for 
{\sc contraction hierarchies}, and generalize a lower bound on the 
number of shortcut edges added in the preprocessing phase. 
Section 5 establishes a lower bound on the query time
of {\sc transit node routing}. We conclude and discuss future
directions in Section 6.

\section{Preliminaries} 

In this paper, we assume
nonnegative integral edge lengths and unique shortest paths.
We will also assume graphs are undirected
in all sections except for the hardness result.
These are standard assumptions to make when proving bounds on
routing algorithms, for example, \cite{hd} and \cite{milo}.

$B_{r}(v)$ represents all nodes $u$ such that $\mbox{dist}(u,v)<r$.  We say a set of nodes \textit{covers}
a set of paths if each path has at least one of its vertices in the set of nodes.

\subsection{Highway Dimension}

Now we will formally define the notion of highway dimension.
\\
\\
\indent
The \textit{highway dimension} of a graph $G=(V,E)$ is the smallest $h$
such that for all $r > 0$ and for all $B_{4r}(v)$, there exists a set $H \subseteq V$,
such that $|H| \leq h$ and $H$ covers all shortest paths of length $\geq r$ in
$B_{4r}(v)$.
\\
\\
\indent
Highway dimension was specifically designed to explain why the best routing algorithms
perform well on real-world graphs but do not perform well on arbitrary graphs.
Although it is too computationally intensive to calculate the exact highway dimension of a continental-sized road network, 
it is conjectured that the highway dimension of
real-world graphs is at most polylogarithmic in the 
number of vertices \cite{hd}.

Abraham et al.\ introduced a slightly refined version of the original highway dimension 
in 2013 \cite{hdnew}.

The difference in the new definition versus the old one is that
instead of having to hit all local shortest paths of length $\geq r$, we have to hit all paths $P$ where
there is a shortest path  $P'$ with endpoints $s$ and $t$ such that $\l(P')>r$, $P \subseteq P'$,
and $P' \setminus P\in\{\emptyset,\{s\},\{t\},\{s,t\}\}$.
That is, we have to hit all paths that can be obtained by removing zero, one, or both endpoints of
a shortest path with length $>r$.  
We will refer to a graph's highway dimension as $h$ for the first definition, and $\hat h$
for the second definition.

The two definitions of highway dimension are very similar but have a few key differences.
Most notably, the new definition bounds the degree of the graph, which was not true before
\cite{hd}.
The new definition of highway dimension allowed
Abraham et al.\ to improve their results on the runtime of routing algorithms.

\subsection{Definition of $G_{t,k,q}$}

Now we will define the family of graphs $G_{t,k,q}$ that will be used in many
of our proofs.
$G_{t,k,q}$ was designed to by Milosavljevi\'c to
show a lower bound on the number of shortcuts created
during the preprocessing stage of  {\sc contraction hierarchies} \cite{milo}.

Consider a complete $t$-ary tree of height $k$ for integers $t,k \geq 2$.
Let $\lambda(v)$ denote the height of node $v$,
and let $\lambda(u,v)$ denote the height of the lowest common 
ancestor between two nodes $u$ and $v$.

Now define the edges
as follows: for all nodes $v$ and $w$ such that $w$ is a proper
 ancestor of $v$, there is an edge between $v$ and $w$
with length $16^{\lambda (w) - 1}$.
This means the edge length from a node $w$ to one of its descendants $v$ is
independent of $\lambda(v)$.  Furthermore, edge lengths increase for nodes higher up in 
the tree.

Denote this graph by $G_{t,k} = (V_{t,k},E_{t,k})$.  
See Figure \ref{fig:gtkq} for an example.
For convenience, we will still refer to this
graph as a tree, even though the additional edges create cycles.

Now we will define $G_{t,k,q} = (V_{t,k,q},E_{t,k,q})$ by taking $q$ copies of $G_{t,k}$,
 and naming them $G_{t,k}^{(a)} =
(V_{t,k}^{(a)},E_{t,k}^{(a)})$ for $a=1,2,...,q$.  The copy of a node $v \in G_{t,k}$ 
in $G_{t,k}^{(a)}$ is denoted $v^{(a)}$.

For all $v \in G_{t,k}$ and $a\neq b$, we add edge
$v^{(a)}$--$v^{(b)}$
to $E_{t,k,q}$ with length $2^{\lambda(v)-k-1}$.
This ensures that switching copies has a low penalty 
($2^{\lambda(v)-k-1}$ is always less
than $1$), and it is always cheaper to switch among copies lower down
in the tree.
See Figure \ref{fig:gtkq} for an example.

\begin{figure}  
\begin{center}  
\includegraphics[height=1.4in,width=4.91in]{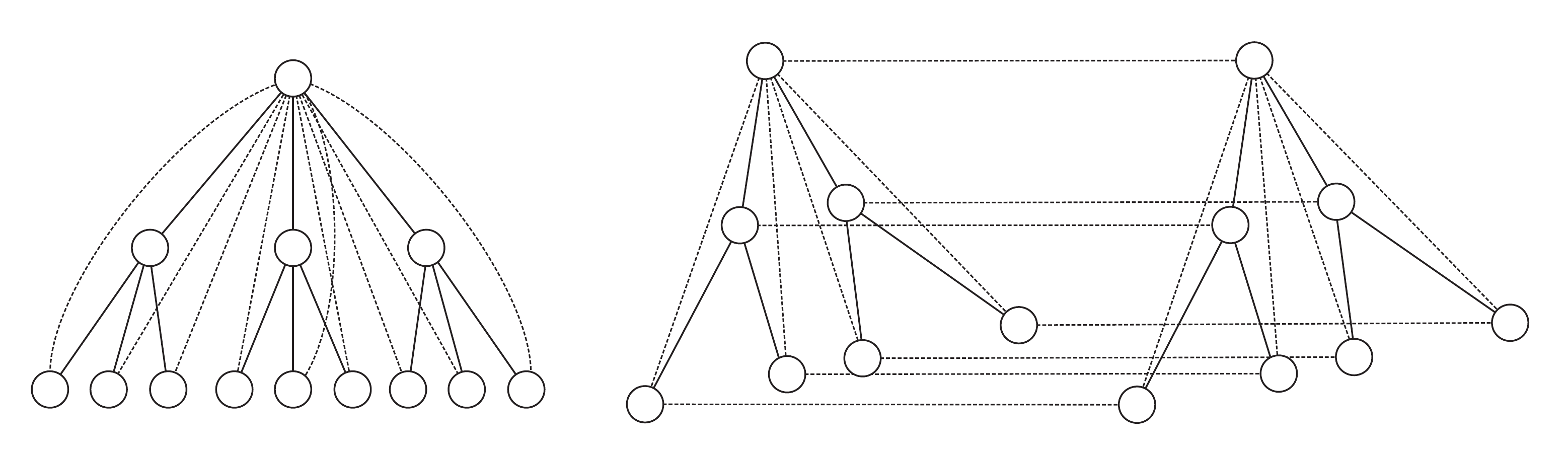}  
\caption{\small \sl The left graph is $G_{3,3}$, and the right graph is
$G_{2,3,2}$ \label{fig:gtkq}}
\end{center}  
\end{figure}

\subsection{Properties of $G_{t,k,q}$}

We will now discuss properties of $G_{t,k,q}$.
The following three lemmas are proven in \cite{milo}.

\begin{lemma} \label{sp tk}
Given $s,t \in V_{t,k}$ with lowest common ancestor $w$, the unique shortest $s$--$t$ path is
$s$--$w$--$t$. \end{lemma}
 
\begin{lemma} \label{sp}
Given $s^{(a)}$ and $t^{(b)}$ in $G_{t,k,q}$, let $w$ be the lowest common ancestor between $s$
and $t$.  Then the shortest $s^{(a)}$--$t^{(b)}$ paths are:
 
$s^{(a)}$--$s^{(b)}$--$w^{(b)}$--$t^{(b)}$, if $\lambda (s) \leq \lambda (t)$, and/or
 
$s^{(a)}$--$w^{(a)}$--$t^{(a)}$--$t^{(b)}$, if $\lambda (t) \leq \lambda (s)$.
\end{lemma}

\begin{lemma} \label{gtkq hd}
The highway dimension $h$ of $G_{t,k,q}$ is equal to $q$, 
the diameter $D$ is $ \Theta (16^k)$, and $|V_{t,k,q}|= \Theta (q t^k)$.
\end{lemma}

It is worth noting that
at the start we assumed graphs have unique shortest paths,
but now many shortest paths in our main family of graphs are not unique.
However, this is a common assumption in routing algorithm proofs because it is not hard to perturb the input
to make all shortest paths unique while maintaining the validity of the proofs.

Additionally, integrality of edge lengths is violated.  
Since the smallest edge is $2^{-k}$
(and all edge lengths are multiples of this),
all of the edge weights can be multiplied by $2^k$ to create integral
lengths.  This will increase
$D$ by a factor of $k$, doubling $\log D$, which will not affect our results.

\section{Hub Labeling} 

The {\sc hub labeling} algorithm was first devised in 2004 by 
Gavoille et al.\ \cite{gav},
and further studied by Cohen et al.\ \cite{cohen}.
However, the algorithm was not practical for continental routing queries until
2011, when Abraham et al.\ came up with an efficient way to perform the preprocessing
and query phases, which made it the fastest routing algorithm to date \cite{hl}.

In this section, we will first give an introduction to the {\sc hub labeling} 
algorithm. Then we will present a lower bound on the query time.
Finally, we will show the preprocessing phase is NP-hard to optimize.

\subsection{The algorithm}

{\sc hub labeling} relies on 
the concept of labeling.  Each node stores information about its shortest paths that allows us
to reconstruct the shortest path during a query.  This idea is used in
a clever way to make queries run very fast.

In the {\sc hub labeling} algorithm, we give each node $v \in V$ a \textit{label}
consisting of other nodes (the \textit{hubs} of $v$), and we store the shortest
distances to the hubs from $v$.  We define a \textit{labeling} $L$ as the 
set of labels $L(v)$ for all $v \in V$.

We construct the
labeling in such a way that for any pair of nodes $s$ and $t$, $L(s) \cap 
L(t)$ contains at least one node on the shortest path from $s$ to $t$.
When satisfied, this is called the \textit{cover property}.
Then in order to perform an $s$--$t$ query, we only need to find the
$v \in L(s) \cap L(t)$ that minimizes 
$\mbox{dist}(s,v)+\mbox{dist}(v,t)$.
This can be made to take $O(|L(s)|+|L(t)|)$ time 
if the labels are sorted with some arbitrary node order.
This process returns $\mbox{dist}(s,t)$.  To return the nodes on this shortest path,
we need to add another data structure in the preprocessing stage, which does
not increase the space complexity by more than a constant factor \cite{hl}.

In Section \ref{hard}, we will show that it is NP-hard to find the labeling
that minimizes the maximum label size for all vertices.
This was suspected to be true.
Therefore, in practice we must rely on heuristics in the 
preprocessing stage.

Abraham et al.\ showed
that the query time of {\sc hub labeling} is $O(\hat h \log D)$, 
using a specific labeling \cite{hdnew}.
The proof did not use any properties of $\hat h$ that are different
from $h$, so we can also say that the query time is $O(h \log D)$.

It is not known how to construct the labeling used in their proof in polynomial time,
so they showed a corollary that uses a polynomial
preprocessing algorithm and permits queries to be handled in $O( h \log h \log D)$ time.

\subsection{Lower bounding the query time}

We cannot prove a lower bound on the minimum query time, since labelings can be constructed
to make any one query run in constant time.
Instead, we will prove a bound on the average query time
by bounding the sum of all label sizes.

\begin{theorem} \label{hl lb}
For all $h$, $D$, $n$, there is a graph $G=(V,E)$ with highway dimension $h$, diameter
$\Theta (D)$, and $|V| \geq n$, such that for any choice of labeling $L$, 
the average query requires $\Omega (h \log D)$ time.
\end{theorem}

\begin{proof}
We will show that $G_{t,k,q}$ satisfies the desired requirements, with $t$, $k$, and $q$
to be defined at the end of the proof.

Consider different classes of shortest paths between 
pairs of leaves distinguished by the height of their lowest common ancestor as follows.

For $0 \leq i \leq k$, let $P_i = 
\{s$--$t \mid s$ and $t$ are leaves, and $\lambda(s,t)=i \}$.

Let $\sum_{v \in V} |L(v)| = H$.  Our goal is to show that a constant fraction of the 
$k+1$ sets $P_0,\;P_1,...,P_k$ each contribute $\Omega
(q^2 t^k)$ distinct nodes to the sum $H$.

We make the assumption that 
all the neighbors of a leaf $v^{(a)}$, and the leaf itself,
are in that leaf's label. That is, $L(v^{(a)})$ contains $v^{(b)}$ for all $b$ 
(even when $b=a$), and contains $w^{(a)}$ for all ancestors $w$ of $v$.
These are $k+q+1$
nodes per leaf and $t^k (k+q+1)$ total nodes, which is asymptotically less
than $\Omega (t^k q^2 k)$,
the desired result.  Therefore, this assumption will not affect the validity of our proof. 

Now consider an arbitrary path in $P_i$.
Label the endpoints of the shortest path $P_i$ by $s^{(a)}$ and $t^{(b)}$.
From Lemma \ref{sp}, $P_i$ must equal 
$s^{(a)}$--$s^{(b)}$--$w^{(b)}$--$t^{(b)}$, where $w$ is the lowest
common ancestor of $s$ and $t$, and $\lambda(w)=i$.

$L(s^{(a)}) \cap L(t^{(b)})$ must contain at least one of $s^{(a)}$, $s^{(b)}$, 
$w^{(b)}$, $t^{(b)}$ in order to satisfy the cover property.  By our 
assumption above, $s^{(a)},s^{(b)} \in L(s^{(a)})$ and
$w^{(b)}, t^{(b)} \in L(t^{(b)})$.  Now there are four cases.

Case 1: $s^{(a)} \in L(t^{(b)})$.  Note that $s^{(a)}$ is not on any other shortest
path starting at $t^{(b)}$.

Case 2: $t^{(b)} \in L(s^{(a)})$.  Again, $t^{(b)}$ is not on any other shortest
path starting at $s^{(a)}$.

Case 3: $s^{(b)} \in L(t^{(b)})$.  $s^{(b)}$ is on all leaf-leaf shortest paths (that end at $t^{(b)}$) of the
form $s^{(c)}$--$s^{(b)}$--$w^{(b)}$--$t^{(b)}$ for $c \neq b$.  There are $q-1$ 
such paths in $P_i$.

Case 4: $w^{(b)} \in L(s^{(a)})$.  $w^{(b)}$ is on all leaf-leaf shortest paths (that start at $s^{(a)}$) of the 
form $s^{(a)}$--$s^{(b)}$--$w^{(b)}$--$v^{(b)}$ for $v$ such that $\lambda (s,v) = i$.  There are $t^i - t^{i-1}$ such paths,
since there are $t^i$ leaves with $w^{(b)}$ as an ancestor, and all but $t^{i-1}$ of those leaves have $w^{(b)}$
as the lowest height ancestor to get to $s^{(a)}$.

Furthermore, 
\begin{equation}
|P_i| = {q \choose 2} t^k (t^i - t^{i-1}) = \frac{q(q-1) t^k (t^i - 
t^{i-1})}{2}
\end{equation}
because there are ${q \choose 2}$ ways to pick two copies of trees, 
$t^k$ choices for the first leaf, and $t^i - t^{i-1}$ choices for the second leaf
(in order to guarantee that the leaves have a lowest common ancestor of height $i$).

So if we assume $t^i-t^{i-1} \geq q-1$ (we will explain in the next paragraph why we can
make this assumption), then we can achieve a lower bound on the number of labels
needed for $P_i$ by exclusively using Case 4 for our choice of labels.
\begin{equation}
\frac{q(q-1) t^k (t^i - t^{i-1})}{2}\div(t^i - t^{i-1})=\frac{q(q-1) t^k }{2}.
\end{equation}

Therefore, the contribution of $P_i$ to the total sum $H$ is at least $\frac{q(q-1) t^k }{2}$.
For all $i$, the hubs that $P_i$ contributes to the sum $H$ have height $i$,
ensuring that a node does not get double counted in $H$. 

Let $k=\lceil \frac {\log D}{4} \rceil$, $q=h$, and pick $t$ big enough such that $q t^{k+1}\geq
n$ (ensuring that $|V| \geq n$) and $t^{k/2} \geq q$ (ensuring that at least half
of the $P_i$'s satisfy $t^i-t^{i-1} \geq q-1$).

Then the highway dimension of $G$ is $h$ and the diameter is $\Theta (D)$.  Recall 
that $|V| \in \Theta(q t^k)$.
Then for any given labeling $L$, 
\begin{equation}
\sum_{v \in V} |L(v)| \geq \frac{k}{2}  \cdot  
\frac{q(q-1) t^k}{2} \in \Omega (h |V| \log D).
\end{equation}

This completes the proof since query times depend on the size of the labels. \qed
\end{proof}

With this theorem, the upper bound presented in \cite{hdnew}
becomes tight.

\subsection{Hardness of preprocessing}

In 2010, Bauer et al.\ established hardness for optimal preprocessing for
a variety of the best routing algorithms, including 
{\sc contraction hierarchies} and {\sc transit node routing} \cite{hard}.
We provide hardness for optimal preprocessing in
{\sc hub labeling} which was suspected to be true \cite{hd}.
By optimal preprocessing, we mean minimizing the maximum hub size over
all vertices.
Babenko et al.\ very recently established hardness for nearly the same
problem,
but they defined optimal preprocessing as minimizing over the \emph{total} label
size \cite{babenko}.
Our definition of optimal corresponds to minimizing the
maximum query time, whereas the other definition corresponds to
minimizing the average query time.

We will switch to directed graphs, which was the original setting of
{\sc hub labeling} \cite{hl}. The main difference is that each node $v$
has a forward label $L_f(v)$ and a reverse label $L_r(v)$,
and the cover property states that for a directed $s$--$t$ query,
$L_f(s)\cap L_r(t)$ is not empty.

Now we formally define the problem {\sc minimum hub labeling} (MHL)
as follows:
\\
\\
\textbf{Problem (MHL).} Given a directed graph $G=(V,A)$ and an integer $k$, 
find a labeling $L$ satisfying the cover property
such that $\max_{v\in V} (\max(|L_f(v)|,|L_r(v)|))\leq k$.
\\
\indent

We will show a reduction from 
a classical NP-hard problem, \textit{exact cover by 3-sets (X3C)}.
In an X3C instance $(U,C)$,
$U$ is a set of elements, 3 divides $|U|$, and $C$ is a set of triples of $U$. 
The problem is whether there exists a set $C'\subseteq C$,
$|C'|=\frac{|U|}{3}$ such that $C'$ covers $U$
(an exact 3-covering of $U$).

Here is an outline of the proof.
Given an X3C instance $(U,C)$, we create an MHL instance
$(G,k)$ where $G=(V,E)$, $U\cup C\subseteq V$ and
for $c\in C$, $u\in U$, $c$--$u\in E$ iff $u\in c$.

We also add a clique of vertices $\{b_1,\dots,b_{2k-1}\}=B$ with arcs to 
nodes in $U$,
whose sole purpose is to fill up the reverse labels of nodes in $U$.
Finally, we add two vertices $\{a_1,a_2\}=A$
with arcs to every node in $C$.

By filling up the reverse labels of nodes $u\in U$,
we force the nodes $a\in A$ to use nodes in $C$ or $U$ for the hubs of
$a$--$u$ shortest paths.
And it is too inefficient to use nodes in $U$ for the hubs,
so nodes in $C$ must act as the hubs.
Then in order for $A$'s label size to stay $\leq k$, there must be an
exact cover for $U$.

\begin{theorem} \label{hard}
Minimum hub-labeling is NP-hard.
\end{theorem}

First we construct a graph $G$, and then prove lemmas about its labeling
until we work up to proving the theorem.

Given an X3C instance $(U,C)$, we create an MHL instance
$(G,k)$ where $G=(V,E)$, 
$V=A \cup C \cup U \cup B$, $|A|=2$, and $|B|=\frac{2}{3}|U|+1$.  
For all $a \in A$ and $c \in C$, there is a directed
edge $(a,c) \in E$.  For all $u \in U$ and $c \in C$ such that $u \in c$, there is
a directed edge $(c,u) \in E$.  For all $b_1,b_2 \in B$ such that 
$b_1 \neq b_2$, $(b_1,b_2)$ and $(b_2,b_1)$ are
in $E$.
Let $B'$ be a subset of $B$ such that $|B'|=\frac{|U|}{3}-1$ 
(it does not matter which $b$'s are in $B'$).
For all $b' \in B'$ and all $u \in U$, there is a directed edge $(b',u) \in E$.
All edges are unit length.
Finally, set $k=\frac{|U|}{3}+1$.
See Figure \ref{fig:maxlabel}.

\begin{figure}  \label{fig:maxlabel}
\begin{center}  
\includegraphics[height=2.5in,width=4.91in]{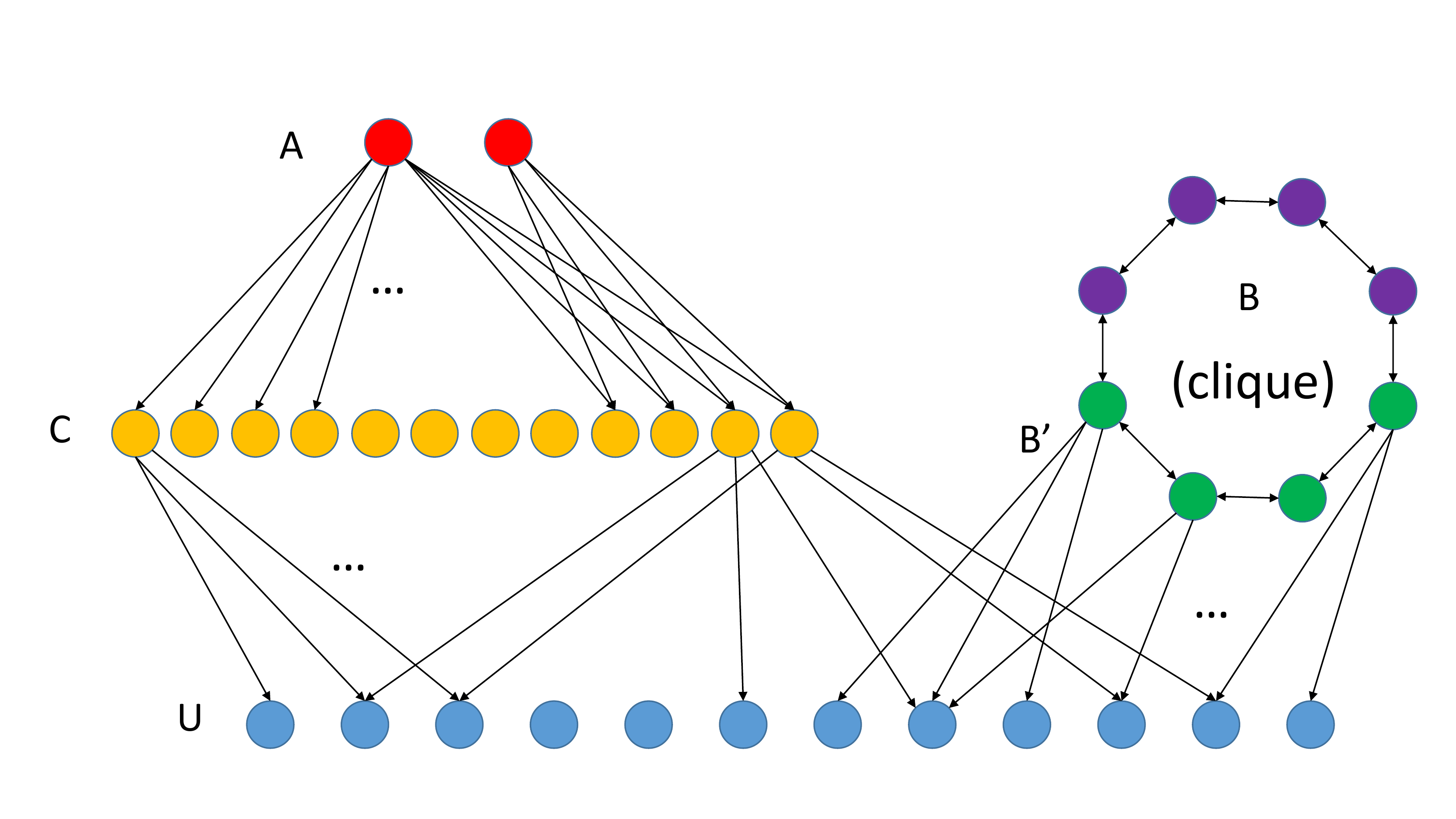}  
\caption{\small \sl The MHL instance constructed from X3C.}
\end{center}  
\end{figure}

First we prove the forward direction: if $(G,k)$ is a {\sc yes}
instance, then $(U,C)$ is a {\sc yes} instance.
We prove this using a few different lemmas.

\begin{lemma} \label{b}
If $(G,k)$ is a {\sc yes} instance,
then for all $b\in B$, $L_f(b)$ and $L_r(b)$ contain $k$ vertices
from $B$.
\end{lemma}

\begin{proof}
Given $b_1$, $b_2 \in B$,
the shortest $b_1$--$b_2$ path is the edge $b_1$--$b_2$, 
since $B$ is fully connected.  
Then to satisfy the cover property, either 
$b_1 \in L_f(b_1) \cap L_r(b_2)$
or $b_2 \in L_f(b_1) \cap L_r(b_2)$.  Each of these cases puts one vertex in 
a label that cannot be reused for any other
shortest path ($b_1  \in L_r(b_2)$ or $b_2 \in L_f(b_1)$).

First we note that for all $b\in B$, $b \in L_f(b)$ and $b \in L_r(b)$.
If this were not the case, (WLOG $b \notin L_f(b)$),
then $L_f(b)$ must contain
$|B|-1=\frac{2}{3}|U|>k$ vertices to satisfy all its outgoing shortest paths,
which contradicts our assumption.

Now we note there are $|B|(|B|-1)$ total shortest paths, 
and each requires adding exactly one node to a label that cannot be reused 
for any other shortest path.
Then the minimum max label we can hope to achieve is $\frac{|U|}{3}+1$,
which corresponds to splitting the $|B|(|B|-1)$ vertices equally among $B$.
So each forward and reverse label has size
$\frac{1}{2}\frac{|B|(|B|-1)}{|B|}=|U|/3$), plus the self hub to reach a total
of $\frac{|U|}{3}+1=k$. \qed
\end{proof}

\begin{corollary} \label{b'}
If $(G,k)$ is a {\sc yes} instance,
then for all $u\in U$, $B'\subseteq L_f(u)$.
\end{corollary}

\begin{proof}
Given $u\in U$, $b'\in B'$, the shortest $u$--$b'$ path is the edge
$u$--$b'$.
Then to satisfy the cover property,
either $u\in L_f(u)\cap L_r(b')$, or $b'\in L_f(u)\cap L_r(b')$.
From Lemma \ref{b}, we know that $L_r(b)$ already contains $k$
vertices from $B$. Therefore, it must be the case that $b'\in L_f(u)$.
Then for all $u\in U$, $B'\subseteq L_f(u)$. \qed
\end{proof}

So, now we know that the reverse labels for nodes in $U$ are almost full up.
To finish off the forward direction, 
we need to show that the only way for vertices in $A$ to have hubs $\leq k$
is to use an exact cover $C'$ for $U$.
Intuitively, it makes sense that the $a$--$c$--$u$ shortest
paths should use the vertices in $C$ as hubs rather than vertices in 
$A$ or $U$, because it can be used for three shortest paths
instead of just one. However, we need to make certain
that some hybrid label with $A$'s, $C$'s, and $U$'s does not work.

Define
$A_1=|\{u\in U\mid u\in L_r(a_1)\}|$ and
$A_2=|\{u\in U\mid u\in L_r(a_2)\}|$.

Also let $U_1=|\{u\in U\mid a_1\in L_r(u)\}|$, and 
$U_2=|\{u\in U\mid a_2\in L_r(u)\}|$.

\begin{lemma} \label{only c}
If $(G,k)$ is a {\sc yes} instance,
then $A_1=A_2=U_1=U_2=0$.
\end{lemma}

\begin{proof}
From Lemma \ref{b'}, we know that for all $u\in U$, 
$L_r(u)$ contains $|B'|=k-2$ vertices from $B'$.
$L_r(u)$ will also need at least one total vertex for all the 
$c$--$u$ shortest paths, for $c$ such that $u\in c$. 
Therefore, we cannot put both $a_1$ and $a_2$ into $L_r(u)$. 

Then, for every $u\in U$ such that $a_1\in L_r(u)$,
$u$ must be in $L_f(a_2)$, or else there would be no other way for the
$a_2$--$u$ path to satisfy the cover property. Therefore, $U_1\leq A_2$.
Similarly, $U_2\leq A_1$.

Now consider $a_1$'s forward label.
$L_f(a_1)$ will need at least one total vertex for all the 
$a$--$c$ shortest paths.
Since there are $A_1$ vertices $u\in U$ such that $u\in L_f(a_1)$,
there is room in $L_f(a_1)$ for $k-1-A_1$ vertices, and we need to
satisfy the cover property for $|U|-U_1$ more shortest paths of the
form $a_1-u$ such that $u\in U$.
The most efficient label for these shortest paths is to pick a vertex
in $C$, which will cover three at a time.
Then we must have $|U|-U_1\leq 3(k-1-A_1)=|U|-3A_1$, from which it follows
that $3A_1\leq U_1\leq A_2$.
With the exact same argument, we get $3A_2\leq A_1$.
Then $9A_1\leq A_1$ and so $A_1=0$.\qed
\end{proof}

\begin{corollary} \label{hard forward}
If $(G,k)$ is a {\sc yes} instance,
then $(U,C)$ is a {\sc yes} instance.
\end{corollary}

\begin{proof}
From Lemma \ref{only c}, it follows that for all $a\in A$ and $u\in U$,
there exists a $c\in C$, such that $u\in c$ and $c\in L_f(a)\cap L_r(u)$.
Then there must be at least $\frac{|U|}{3}$ vertices from $C$ in $L_f(a)$.
$L_f(a)$ also needs a hub for all $a$--$c$ shortest paths where $c\in C$.
The only way to accomplish that is to let $a$ be the hub.
Then $L_f(a)$ contains $a$, plus some $C'\subseteq C$ such that for all
$u\in U$, there exists a $c\in C'$ such that $u\in c$.
Since $(G,k)$ is a {\sc yes} instance, $|C'|\leq k-1=\frac{|U|}{3}$.
But then $C'$ is an exact cover for $U$, so 
$(U,C)$ is a {\sc yes} instance.\qed
\end{proof}

Now we will show the backward direction.
Proving the forward direction alludes to a specific labeling,
so now it is just a matter of showing this labeling is actually possible.

\begin{lemma} \label{hard backward}
If $(U,C)$ is a {\sc yes} instance, then $(G,k)$ is a {\sc yes} instance.
\end{lemma}

\begin{proof}
Let $C'\subseteq C$ be an exact cover for $C$.
Given $u\in U$, denote $c'_u$ as the element in $C'$ that covers $u$.
We present the following labeling $L$.

For $a\in A$, $L_f(a)=\{a\}\cup C'$, and $L_r(a)=\emptyset$.

For $c \in C$, $L_f(c)=\{u\mid u\in c\}$ and $L_r(c) = A$.

For $u \in U$, $L_f(u) = \emptyset$ and 
$L_r(u) = \{u,c'_u\} \cup B'$.

For $b_i \in B$, $L_f(b_i) = \{b_i,b_{i+1},\dots,b_{i+k \mod 2k}\}$ and \\
$L_r(b) = \{b_i,b_{i-1},\dots,b_{i-k\mod 2k}\}$.

It is easy to check this labeling satisfies the cover property.
Each $a$--$c$ shortest path uses $a$ as a hub.
Each $a$--$u$ shortest path uses $c'_u$ as a hub.
Each $c$--$u$ shortest path uses $u$ as a hub.
Each $b'$--$u$ shortest path uses $u$ as a hub.
Given $b_i,b_j\in B$. If $j\leq i+k \mod 2k$, then the $b_i$--$b_j$
shortest path uses $b_j$ as a hub. Otherwise, it uses $b_i$ as a
hub.

Also, it it clear that every label has size $\leq k$.
This completes the proof.\qed
\end{proof}

The proof of Theorem \ref{hard} follows directly from
Corollary \ref{hard forward} and Lemma \ref{hard backward}.

\section{Contraction Hierarchies}

{\sc Contraction hierarchies} \cite{ch} is a shortcut-based algorithm, making it fundamentally
different from {\sc hub labeling}.
It works by running bidirectional Dijkstra search, pruning the searches
based on a node's importance.

In this section, we explain how the {\sc contraction hierarchies} algorithm works, 
prove a lower bound on the query time, and then generalize a result about
the number of shortcut edges added in the preprocessing phase.

\subsection{The algorithm}

In the preprocessing stage for {\sc contraction hierarchies}, we iteratively \textit{contract}
nodes using a predefined ordering, called a 
\textit{contraction ordering}.  
The contraction operation called on $v$ first deletes $v$ from the graph,
and then may add edges between $v$'s neighbors if they are needed 
to preserve the shortest path lengths. Any such edge is put into a set $E^+$.
We contract every node in the graph based on the ordering, and we are left
with the set $E^+$ of ``shortcut edges''.

To run an $s$--$t$ query, run bidirectional Dijkstra search from $s$ and $t$
on the graph $G^+=(V,E \cup E^+)$, 
except at node $v$, only consider edges $v$--$w$ in which $w$ 
was contracted after $v$.  When there are no more nodes to consider in either direction, find the
node $v$ that minimizes the sum of its distances to $s$ and to $t$.  

In \cite{ch}, it is proven that $v$ is guaranteed to be on the shortest path
between $s$ and $t$, which means that $\mbox{dist}(s,t)=\mbox{dist}(s,v)+\mbox{dist}(v,t)$,
so the query returns the shortest $s$--$t$ path.

Note that any contraction ordering
will give correct queries, 
but a better contraction ordering will make $|E^+|$ small,
decreasing time and space requirements.
Finding the optimal
ordering is NP-hard \cite{hard}, but there are fast heuristics that make $|E^+|$ 
within $\log h$ of optimal \cite{hdnew}.

Abraham et al.\ showed an upper bound on the query time of 
{\sc contraction hierarchies} 
that depends on $\Delta$:
$O((\Delta+h \log D)(h \log D))$ \cite{hd}.
Using the new definition of highway dimension, 
Abraham et al.\ achieved the better bound of
$O((\hat h \log D)^2)$ time.
Both of these assume optimal preprocessing.
If a polynomial time preprocessing algorithm is required, the bounds are modified
to $O((\hat h \log \hat h \log D)^2)$ and $O((\Delta+h \log h \log D)(h \log h \log D))$.

\subsection{Lower bounding the query time} \label{chlb}

We show a lower bound using the old definition of highway dimension.

\begin{theorem} \label{ch lb}
For all $h$, $D$, $n$, there is a graph $G=(V,E)$ with highway dimension $h$, diameter
$\Theta (D)$, and $|V| \geq n$ such that the average query time is
$\Omega ((h \log D)^2)$ for {\sc contraction hierarchies}.
\end{theorem}

Our strategy will be to find a lower bound assuming Abraham et al.'s
(optimal) ordering, and then show that modifying the ordering can only increase the runtime.

\cite{milo} provided a criterion for shortcut paths
in the optimal ordering:
the path $s^{(a)}$--$s^{(b)}$--$w^{(b)}$--$t^{(b)}$ is shortcut 
if and only if $a \neq b$, $w$ is a proper
ancestor of $s$, and $s^{(b)}$ is contracted before $s^{(a)}$.
First we present a proof sketch, and then we give the formal proof.

Here is an outline of the proof.
Again we will use $G_{t,k,q}$, and we limit our analysis to leaf-leaf
queries, which make up the majority of all queries.
First we prove the theorem assuming Abraham et al.'s contraction order.
For $G_{t,k,q}$, this means nodes are contracted based on their height
in the tree.
In the forward search of a leaf-leaf query 
$s^{(a)}$--$t^{(b)}$,
the only nodes we may visit are ancestors $v^{(c)}$ of $s$ 
such that $v^{(c)}$ is contracted after $v^{(a)}$.
Then half of these nodes will have lower contraction order than the other half,
and so it can be shown that the shortcut criterion guarantees $\Omega(q^2 k^2)$
edges will be created along half of the forward searches.

Then we show that veering away from this ordering will
only increase the number of shortcut edges produced (or slightly decrease,
but not by more than a constant factor).
This is more technical. The main idea is to carefully
examine the effects of contracting a node higher up in the tree, before all of
its descendants were contracted. Although contracting a higher node $v$ decreases
some of the paths from any descendant $u$ to $v$, it creates shortcuts between
all pairs of descendants which have not yet been contracted, which could
cause an exponential number of extra edges to be created. The overall
difference does not increase the big-Omega bound from Abraham et al.'s
contraction order.

\begin{proof}

We will show $G_{t,k,q}$ satisfies the properties, defining $t,k,q$ at the end of
the proof.
Consider a query between two leaves $s^{(a)}$ and $t^{(b)}$ such that 
$\lambda(s,t)=k$ and $a \neq b$.  This type of query makes up a constant fraction of all
queries, so we will limit our analysis to this case.
A regular Dijkstra search settles $s^{(a)}$ and all copies of $s$, and then it settles the parent of
$s^{(a)}$ and all its copies, and continues to settle the successive ancestors of $s^{(a)}$
along with their copies.  A total of $q(k+1)$ nodes are settled in this way.  The backwards search goes through a similar process starting at $t^{(b)}$.
For {\sc contraction hierarchies}, each node only needs to look at neighbors with a higher contraction
order than itself.  If we are using an adjacency list to represent the graph, this can be done by
reordering the adjacency list based on contraction order.

Assume initially that we are using Abraham et al.'s contraction ordering, 
which orders nodes by height from the bottom up (we will remove this assumption shortly).
So in the forward search, the only nodes we may visit are ancestors of $s$ (in any copy).
We refer to this set of nodes as $S$ and recall that it contains $q$ nodes at each layer
$i$ in the tree.  Among the nodes in $S$ with height $i$, let $T_{i}$ contain the
$\frac{q}{2}$ nodes with lower contraction order than the other $\frac{q}{2}$ nodes in that
layer.  Let
\begin{equation}
T=\bigcup _{i=0}^{k/2} T_i.
\end{equation}

Suppose $v^{(c)}$ is one of the $\frac{qk}{4}$ nodes in $T$.
Recall that the shortcut criterion for Abraham et al.'s ordering says
the path $v^{(d)}$--$v^{(c)}$--$u^{(c)}$ (where $u$ is an ancestor of $v$) will be shortcut if $v^{(c)}$ is
contracted before $v^{(d)}$.  Then contracting $v^{(c)}$ will create at least $\frac{qk}{4}$ shortcuts, 
since $v^{(c)}$ is in
the bottom half of the tree and has a lower contraction order than half of the nodes $v$ in other copies.
Therefore, the
forward search will need to look through at least $\frac{q^2 k^2}{16}$ nodes, making the average
query take $\Omega(q^2 k^2)$ time.

Now we will consider a general ordering
by examining the effects of contracting an arbitrary node $v^{(a)}$ on edges in $G_{t,k}^{(a)}$.

If $v^{(a)}$ is contracted before a descendant $u^{(a)}$, shortcuts from $u$ in any copy to $v^{(a)}$ will never
be created.  The number of queries this affects is based on the height of $u$.  If $u$ is a leaf, it only affects queries starting from $u$, but if $u$ is higher up in the tree, it will affect all queries starting at
leaves with $u$ as an ancestor.  In effect, we need to weight the nodes based on their importance.  
We do this using $\sum_{u^{(a)}} t^{\lambda(u)}$, where $u^{(a)}$ is a descendant of $v^{(a)}$ with contraction
order higher than $v^{(a)}$.  
The value of this sum is proportional to the loss in total query time when contracting $v^{(a)}$ compared to
Abraham et al.'s ordering. 
Let $\lambda(v)=i$.
If all of $v^{(a)}$'s descendants were contracted before $v^{(a)}$,
the sum would be $\sum_{j=0}^{i-1} t^{i-j} t^j = i t^i$ because each layer can contribute at most $t^i$ to the sum.
There are two cases to consider.

Case 1: $\sum_{u^{(a)}} t^{\lambda(u)} \leq \frac{1}{2} i t^i$.
In this case, the average query time decreases by at most a factor of two, 
which doesn't affect our big-Omega bound.

Case 2: $\sum_{u^{(a)}} t^{\lambda(u)} > \frac{1}{2} i t^i$.
The number of edges in $G_{t,k}^{(a)}$ that are lost from $v^{(a)}$'s contraction is $\leq t^i$, the number of $v^{(a)}$'s descendants.
However, contracting $v^{(a)}$ before many
of its descendants will create many leaf-leaf shortcuts.

The smallest possible set of contracted descendants would contain the $\geq t^{i/2}$ nodes
in the top $\frac{i}{2}$ layers below $v^{(a)}$.

Given two of these nodes $x^{(a)}$ and $y^{(a)}$ with $\lambda(x,y)=\lambda(v)$,
a shortcut will be created between $x^{(a)}$ and $y^{(a)}$.  Half of the subtrees
rooted at $v^{(a)}$'s children
will have half of their nodes with contraction order higher than $v^{(a)}$, so we will gain at least
${t/2 \choose 2} (\frac{t^{i/2}}{2t})^2 = \frac{t^{i-1}(t/2-1)}{16} \in \Omega(t^i)$ extra shortcuts this way.

Therefore, the number of edges decreases by at most a constant factor, which does not affect
our big-Omega bound.

In both cases, we maintain the 
$\Omega(q^2 k^2)$ bound even with an arbitrary ordering.

Now let
$k=\lceil \frac{\log D}{4} \rceil$ and $q=h$, and we pick $t$ big enough
such that $q t^{k} \geq n$. 
Then the average query for {\sc contraction hierarchies} is $\Omega ((h \log D)^2)$.
\qed
\end{proof}

\subsection{Lower bounding the size of $E^+$}

Abraham et al.'s upper bound of $O((\hat h \log D)^2)$ on the query
time
involves proving that $|E^+| \in O(n \hat h \log D)$.
The latter bound was proven tight in \cite{milo}.
However, the proof assumes the contraction order from the algorithm in Abraham et al.\
which is thought to be NP-hard to compute.  We show a new proof of this 
lower bound generalized to any contraction order.

\begin{theorem} \label{ch preprocess}
For all $h$, $D$, $n$, there is a graph $G=(V,E)$ with highway dimension $h$, diameter
$\Theta (D)$, and $|V| \geq n$ such that for any contraction
ordering, $|E^+| \in \Omega (h |V|
\log D)$.
\end{theorem}

\begin{proof}
We will show $G_{t,k,q}$
satisfies the desired requirements, setting the values of $t,k,q$ at the end of the proof.

We will be concerned only with shortcuts added when contracting leaves. 
We will first count the number of shortcuts added by contracting all of the leaves
first, as in the preprocessing algorithm by
Abraham et al.
Recall the criterion for creating a shortcut in this ordering, which was stated in Section \ref{chlb}.
A path $s^{(a)}$--$s^{(b)}$--$w^{(b)}$--$t^{(b)}$ is shortcut if and only if $a \neq b$, $w$ is a proper
ancestor of $s$, and $s^{(b)}$ is contracted before $s^{(a)}$.  Then the number of 
shortcuts added when contracting all of the leaves is $S =t^k (k) {q \choose 2} \in 
\Theta (qk|V|)$ since there are $t^k$ ways of picking a leaf, $k$ ways of picking a 
proper ancestor, and ${q \choose 2}$ ways of picking two copies.

In general, the number of shortcuts created for leaf $v^{(a)}$ at the time of its contraction
is the number of ancestors
$v^{(a)}$ has in $G_{t,k}^{(a)}$ multiplied by the number of copies $v^{(b)}$, 
$b \neq a$, in other trees.  We will now consider the effects of arbitrary contraction
order on the number of edges a leaf has in its own copy at its time of contraction.

Given an arbitrary contraction order $\theta$ and a non-leaf $v$, let $c_i$, $1 \leq i
\leq t$, be the number of leaves with contraction order higher than $v$ in the
subtree with $v$'s $i$th child as a root.  Then $0 \leq c_i \leq t^{\lambda(v)-1}$
for all $i$.

Contracting $v$ causes $\sum_{i=1}^k c_i$ leaf descendants of $v$ to lose one edge each.
However, contracting $w$ also increases the number of leaf-leaf edges by $\sum_{i \neq
 j} c_i c_j$.

Then the net edge gain for contracting $v$ instead of all leaves first is
\begin{equation}
A_{v,\theta} = \sum_{i \neq j} c_i c_j - \sum_{i=1}^t c_i.  
\end{equation}
In order to find the minimum value of $A_{v,\theta}$, we consider four cases.

Case 1: $\geq 3$ $c_i$'s are nonzero.  Without loss of generality, let the $c_i$'s make
a decreasing sequence.  So $c_1 \geq c_2 \geq \cdots \geq c_t \geq 0$ and $c_3 \geq 1$.
Then 
$c_1 c_3 \geq c_1,\;c_1 c_2 \geq c_2, \;c_2 c_3 \geq c_3,...,\;c_{t-1} c_t
\geq c_t$.
It follows that
\begin{equation}
A_{v,\theta} =  \sum_{i \neq j} c_i c_j - \sum_{i=1}^t c_i \geq 0.
\end{equation}

Case 2: Exactly two $c_i$'s are nonzero.  So $c_1 \geq c_2 \geq 1$ and $c_3 = c_4 =
\cdots = c_t = 0$.  Then $A_{w,\theta}= c_1 c_2 - c_1 - c_2 = (c_1 - 1)(c_2 - 1) - 1$.  If
$c_2 > 1$, then $(c_1 - 1) \geq (c_2 - 1) \geq 1$, so $A_{v,\theta} \geq 0$.  
If $c_2 = 1$, then $c_1 - 1 = 0$, so $A_{v,\theta} = -1$.

Case 3: Exactly one $c_i$ is nonzero.  So $c_1 \geq 1$ and $c_2 = c_3 = \cdots = c_t =
0$.  Then $A_{v,\theta} = -c_i$, so the minimum value of $A_{v,\theta}$ in this case is 
$-t^{\lambda(v)-1}$.

Case 4: All $c_i$'s are zero.  Then clearly $A_{v,\theta} = 0$.

Therefore, the minimum value of $A_v$ is $-t^{\lambda(v)-1}$ from case 3.

Note that the possible leaf-leaf edges we gain from contracting a non-leaf $v$ are
independent of other leaf-leaf edges we gain from contracting another non-leaf $u$:
if $\lambda(v) = \lambda(u)$, the leaves in the edges must be different since they
cannot have both $v$ and $u$ as an ancestor.  If $\lambda(v) \neq \lambda(u)$, the edges
must be different since the lowest common ancestors between the endpoints of each
edge are at different heights.

So given an arbitrary contraction order, the number of leaf-edges within a copy $G_{t,k}^{(a)}$
(at the time of the leaf's contraction) is
\begin{equation}
k t^k - \sum_{v \in G_{t,k}^{(a)}} A_{v,\theta} \geq k t^k - \sum_{i=1}^k t^{k-i} t^{i-1} = k t^k - \sum_{i=1}^k 
t^{k-1} = k t^k - k t^{k-1} = k t^{k-1} (t-1). 
\end{equation}

Then 
\begin{equation}
|E^+| \geq {q \choose 2} k t^{k-1} (t-1) \in \Omega (kq|V|)
\end{equation}

We let $k=\lceil \frac{\log D}{4} \rceil$ and $q=h$, and we pick $t$
such that $q t^{k+1} \geq n$.  Then $G$ has highway dimension $h$, diameter 
$\Theta (D)$, and has $|V| \geq n$.  Finally, given a contraction order $\theta$, 
$|E^+| \in \Omega (h |V| \log D)$. \qed
\end{proof}

\section {Transit Node Routing}

{\sc transit node routing} \cite{tnr} was devised in 2007 by Bast et al., and it (and variants) remain the 
second-fastest family of routing algorithms, behind {\sc hub labeling}
\cite{rptn}.  However, {\sc transit node routing} requires about an order of 
magnitude less space than {\sc hub labeling}.
In this section, we first review the {\sc transit node routing} algorithm, and then
we give a lower bound on the query time.

The algorithm works by picking a set $T \subset V$ of \textit{transit nodes} that hits many long-distance shortest paths. 
$|T|$ is often chosen to be in $\Theta(\sqrt{|V|})$, which makes the algorithm run fastest while
maintaining that additional memory requirements are bounded by the input graph size.
Usually, the contraction order is used to pick $T$ (since contraction order essentially seeks to measure a node's importance
with respect to shortest paths), which works well in practice.

Next, given any node $v$, $A(v) \subset T$ is the set of that node's \textit{access nodes}, which
are chosen to hit the long-distance queries stemming from $v$. This usually means that we want to pick nodes 
in $T$ that are close to $v$.

The distances between all pairs of transit nodes are computed and stored, 
as well as the distances between a 
node $v$ and each of its access nodes.
A query is called a \textit{global query} if $\min (\mbox{dist}(s,u)+\mbox{dist}(u,v)+\mbox{dist}(v,t)\;|
\;u \in A(s),\;v \in A(t)) = \mbox{dist}(s,t)$.  Otherwise, it is a \textit{local query}.
To run an $s$--$t$ query, first run a quick locality filter that determines whether the query is local. This filter is allowed to make one-sided errors;
it can misclassify a global query as local, but not the other way around.
Locality filters are historically calculated using
the coordinates of the vertices.
If it is a global query, calculate the minimum $\mbox{dist}(s,u)+\mbox{dist}(u,v)+\mbox{dist}(v,t)$ by trying all combinations
of access nodes from $A(s)$ and $A(t)$.  Local queries are handled by a fast local search 
such as {\sc contraction hierarchies}.

Abraham et al.\ use a choice of $T$ based on multiscale shortest-path covers to prove that access 
nodes are bounded in size by $O(\hat h)$,
from which it follows that global queries can be handled in $O(\hat h^2)$ time.
Local queries done using {\sc contraction hierarchies} can be handled in $O((\hat h \log D)^2)$ time
as we saw in the previous section (however, local queries tend to be small, making the queries run much faster than the average {\sc contraction hierarchies} query).

This bound is not possible without the new definition of highway dimension.
Again, if we want polynomial time preprocessing, the query time bound for global queries 
increases to $O((\hat h \log \hat h)^2)$.

\subsection{Lower bounding the query time}

While the upper bound for {\sc transit node routing} was for global queries only,
our lower bound will include both local and global searches.
We will use {\sc contraction hierarchies} for local queries.

\begin{theorem} \label{tnr lb}
For all $h$, $D$, $n$, there is a graph $G=(V,E)$ with highway dimension $h$, diameter
$\Theta (D)$, and $|V| \geq n$ such that for any choice of transit nodes $T$
and access nodes $A$, the average query time is $\Omega(h^2)$.
\end{theorem}

We work up to the proof of Theorem \ref{tnr lb} using a series of definitions
and lemmas.

Call a leaf-leaf shortest path \textit{regular} if the shortest path is global and
neither endpoint is a transit node.  We would like to exclude irregular shortest
paths from our analysis.

First, we show that queries with a transit node as an endpoint do not make up
a constant fraction of all queries.  Since $|T| \leq \sqrt{|V|}$, the number of 
shortest paths
in which at least one endpoint is a transit node is $O(|V| \sqrt{|V|}) \in o(|V|^2)$.

Next, we consider the case in which local queries make up at least a
$\frac{1}{16}$ fraction of total queries.  In the previous section, we showed in Theorem \ref{ch lb}
that the average query
for {\sc contraction hierarchies} requires $\Omega((h \log D)^2)$ time.
The proof showed a constant fraction of all queries required this amount of time.
If we lower all of the constants in the proof, we can show that given any set of
$\frac{1}{16}$ of total queries, a constant fraction of those queries require
$\Omega((h \log D)^2)$ time (thus a constant fraction of all queries require that amount
of time).  This big-Omega bound is higher than the one we seek to prove
for global queries, so for the rest of our analysis we can assume that a
$< \frac{1}{16}$ fraction of total queries are irregular.
In particular, this means a constant fraction of the total queries are 
regular.

There is a simple intuition for the rest of the proof. Given a regular 
shortest
path $s^{(a)}$--$t^{(b)}$, either $s^{(a)}$ or $t^{(b)}$ must have an access node
in the other's copy, since the non-endpoint vertices on the shortest path all come
from one copy.
The proof becomes technical because we must show that a constant fraction of
leaves have a large amount of access nodes in distinct copies and subtrees.
But we are able to show that a constant fraction of the nodes need 
$\Omega(q^2)$
access nodes, and the proof follows.

\begin{lemma} \label{halfs}
If the number of local leaf-leaf queries is a $<\frac{1}{16}$ fraction of total queries, then
there is a set of $\frac{q}{2}$ copies in which there are $\frac{t^k}{2}$ 
leaves that are each an endpoint of
$\frac{t^k}{2}$ regular shortest paths going to
at least $\frac{q}{2}$ different copies.
\end{lemma}

\begin{proof}
Assume the number of local leaf-leaf queries is $o(|V|^2)$, but assume the lemma is false.
Then there must be $\geq \frac{q}{2}$ copies with the following
property:  $\geq \frac{t^k}{2}$ leaves are each endpoints of $\leq \frac{t^k}{2}$ regular shortest paths 
going to $\geq \frac{q}{2}$ copies.

Now consider the maximum number of regular leaf-leaf shortest paths possible in $G_{t,k,q}$ under
that assumption.  Making all four inequalities tight, we have $\frac{q}{2}$ copies with $\frac{t^k}{2}$ leaves each as 
endpoints of $\frac{t^k}{2}$ regular shortest paths going to $\frac{q}{2}$ copies each.  In other words, 
in half of the copies, half of the leaves each have the property that in half of the copies,
half of the shortest paths going from that leaf to the copy are regular.
This means that at the very least $\frac{1}{2} \cdot \frac{1}{2} \cdot \frac{1}{2} \cdot \frac{1}{2}
=\frac{1}{16}$ of all leaf-leaf shortest paths must not be regular.  

This violates one of our assumptions, so we have a contradiction. \qed
\end{proof}

Now we have the machinery necessary to prove Theorem \ref{tnr lb}.

\begin{proof}

We will show that
$G_{t,k,q}$ has the desired properties, with the values of $t$, $k$, and $q$ to be defined
at the end of the proof.

From our previous argument at the start of this subsection, we need only consider the case where $<\frac{1}{16}$ of all queries are irregular.

We use Lemma \ref{halfs} to define a set $S$ of regular shortest paths such that there
are exactly $\frac{q}{2}$ copies that have exactly $\frac{t^k}{2}$ leaves with exactly
$\frac{t^k}{2}$ regular shortest paths in $S$ going to $\frac{q}{2}$ copies.

Then 
\begin{equation}
|S| \geq \frac{1}{2} \cdot \frac{q}{2} \cdot \frac{t^k}{2} \cdot \frac{q}{2} \cdot \frac{t^k}{2}=
\frac{q^2 t^{2k}}{32}.  
\end{equation}
We added another factor of $\frac{1}{2}$ because these shortest paths can be double counted.

Given a path $P \in S$, $P$'s endpoints are two leaves $s^{(a)}$ and $t^{(b)}$ in 
different copies and must be of the form 
$s^{(a)}$--$s^{(b)}$--$w^{(b)}$--$t^{(b)}$ or 
$s^{(a)}$--$w^{(a)}$--$t^{(a)}$--$t^{(b)}$ by Lemma \ref{sp}.  Without loss of generality, assume that $P$ is $s^{(a)}$--$s^{(b)}$--$w^{(b)}$--$t^{(b)}$.
Since the path is global, $s^{(a)}$ must have an access node on $P$.  The access node can't be $s^{(a)}$ itself since
$P$ is regular.  Therefore, the access node must be in $G_{t,k}^{(b)}$.

This access node hits at most $\frac{t^k}{2}$ paths 
in $S$ stemming from $s^{(a)}$ because that is the total number of shortest paths in $S$ from $s^{(a)}$ 
to a leaf in $G_{t,k}^{(b)}$.

So given an arbitrary path in $S$, we have shown that an access node for some node $v^{(a)}$
must exist that can hit at 
most $\frac{t^k}{2}$ other shortest paths in $S$.  Then the total number of access nodes needed 
in $S$ is at the very least
\begin{equation}
\frac{q^2 t^{2k}}{32} \div \frac{t^k}{2} = \frac{q^2 t^k}{16} \in \Omega (q |V|).
\end{equation}

As in earlier proofs, we let $k=\lceil \frac{\log D}{4} \rceil$ 
and $q=h$, and we pick $t$
such that $q t^{k+1} \geq n$.  Then $G$ has highway dimension $h$, diameter 
$\Theta (D)$, and has $|V| \geq n$.  

Queries in which both endpoints' access node sets are $\Omega(h)$ will take $\Omega(h^2)$ time,
and these make up a constant fraction of all global queries. \qed
\end{proof}

\section{Conclusions and Future Work}

We proved lower bounds on the query time of 
{\sc hub labeling}, {\sc contraction hierarchies}, and {\sc transit node routing}.
The proofs are all quite different, despite using the same family of graphs
for each proof.
We also generalized a lower bound on the size of $E^+$ in {\sc contraction hierarchies}
preprocessing, and established hardness for optimal preprocessing in {\sc hub labeling}.

Although we have proven lower bounds for the query times of three state-of-the-art
algorithms, the graphs used in the arguments are not representative of real-world graphs.
For instance, the graphs do not have small separators and are not planar.
This implies it may be possible to circumvent this lower bound using different properties
that better capture the structure of real-world graphs.

Another way to work with more realistic road networks is to use the idea of multiscale
dispersed graphs, defined in \cite{mdg}, as a new model for graphs that simulate real-world
graphs.  One may be able to obtain better bounds on the query time with this model.

Throughout this paper, we assumed undirected graphs, so future work could extend these results
to the directed case.
Furthermore, apart from {\sc hub labeling}, the upper and lower bounds are not tight because
of the different definitions of highway dimension.
Ideally, we would find a way to prove the lower bounds using the more recent definition
of highway dimension.  However, we cannot use $G_{t,k,q}$ for this task.  Under the new
definition, $G_{t,k,q}$ has highway dimension at least $q+k$, since the new definition
guarantees a graph's degree is bounded by its highway dimension.

\section*{Acknowledgments}\label{sec:Acknowledgments}

The results in this paper are from the senior honors 
thesis of the author, written under the
direction of Prof. Lyle McGeoch, at Amherst College.  We would like to
give a huge thanks to
Lyle McGeoch for helpful discussions and suggestions throughout the writing process.
We are grateful for the Post-Baccalaureate Summer Research Fellowship program
at Amherst College, which supported the writing of this paper.

\bibliographystyle{plain}
\bibliography{routing}

\end{document}